\documentclass[conference,10pt]{IEEEtran}
\usepackage{amssymb,amsmath}
\usepackage{mathrsfs}
\usepackage{graphicx}
\usepackage{multirow}
\usepackage{amssymb,amsmath, epsfig}
\psfull
\newcommand{\vect}[1]{{\mathbf{#1}}}
\newtheorem{theorem}{Theorem}[section]
\newtheorem{lemma}[theorem]{\bf {Lemma}}
\newtheorem{therm}[theorem]{\bf {Theorem}}

\newenvironment{definition}[1][Definition]{\begin{trivlist}
\item[\hskip \labelsep {\bfseries #1}]}{\end{trivlist}}

\newenvironment{remark}[1][Remark]{\begin{trivlist}
\item[\hskip \labelsep {\bfseries #1}]}{\end{trivlist}}

\newcommand{\qed}{\nobreak \ifvmode \relax \else
      \ifdim\lastskip<1.5em \hskip-\lastskip
      \hskip1.5em plus0em minus0.5em \fi \nobreak
      \vrule height0.75em width0.5em depth0.25em\fi}

\begin{document}

\title{
Distributed Optimization of Multi-Cell Uplink Co-operation with Backhaul Constraints}
\author{Shirish Nagaraj  \\
Technology \& Innovation - Research \\
Nokia Networks \\Arlington Heights, IL 60004 \and Michael L. Honig, Khalid Zeineddine \\
Dept.\ of EECS
\\Northwestern University
\\Evanston, IL 60208
}

\maketitle

\begin{abstract}
We address the problem of uplink co-operative reception
with constraints on both backhaul bandwidth and the receiver aperture,
or number of antenna signals that can be processed.
The problem is cast as a network utility (weighted sum rate) maximization
subject to computational complexity and architectural 
bandwidth sharing constraints. 
We show that a relaxed version of the problem is convex,
and can be solved via a dual-decomposition.
The proposed solution is distributed in that each cell broadcasts 
a set of {\em demand prices} based on the data sharing requests they receive. 
Given the demand prices, the algorithm determines an antenna/cell ordering 
and antenna-selection for each scheduled user in a cell. 
This algorithm, referred to as {\em LiquidMAAS}, iterates 
between the preceding two steps. 
Simulations of realistic network scenarios show that the algorithm 
exhibits fast convergence even for systems with large number of cells.  
\end{abstract}


\section{Introduction}
\label{intro}
Uplink Co-ordinated Multi-Point (CoMP) is a promising technique
for increasing the capacity of 4G networks \cite{fettwiss,Karakayli,Patero,Liao}. 
The uplink is gaining increasing attention due to the dramatic increase
of user-generated data in the form of photos, videos and file-sharing. 
In practice, sharing of uplink received signals across cells is limited
by backhaul bandwidth. In addition, the receiver {\em aperture},
or number of signals from antennas at neighboring cell sites
that can be processed at a particular cell, may be limited due 
to hardware constraints.

Several approaches for signal sharing and combining have been proposed 
for uplink CoMP. A performance analysis of different combining methods
with different backhaul bandwidth requirements
is presented in \cite{fettwiss_marsh}.
To reduce the amount of sharing, dynamic clustering of cells has
been proposed, e.g., in \cite{Papadogiannis,giannakis}.
In addition to limiting the amount of information that can be
shared across cells, the limited backhaul bandwidth also introduces
latency, which is addressed in \cite{dp}.

Previous work on CoMP has generally assumed that the set of cells
that share information is fixed {\em a priori} by the topology and 
the bandwidth of the backhaul links. In this paper we relax this assumption
and consider the problem of {\em optimizing} {\em sets} of {\em helper}
cells that pass along their uplink signals to other cells.
(A cell can both share its signals as a helper cell
and receive signals, or help from other cells.)
We account for architectural constraints that limit the set of potential
helper cells (which differs from cell to cell), 
backhaul constraints that limit the number of cells
to which a helper cell can send its received signals ({\em egress} constraint), 
and hardware constraints at the cell site, which may limit
the number of incoming signals the cell can combine ({\em ingress} constraint). 

Related work in \cite{giannakis,weiyu}  has considered the problem of MMSE receiver estimation under compression and backhaul constraints. We focus here instead on a simpler sharing formulation that  introduces explicit constraints on egress bandwidth, to arrive at a convex weighted sum rate formulation with guaranteed convergence. Imposing these explicit egress bandwidth and ingress aperture constraints captures important architectural limitations in centralized or distributed co-operative networks.
Our problem is then to maximize a sum rate objective subject
to these architectural constraints. Since the cell sites are assumed to
have multiple antennas, we refer to this problem as
{\em Multi-Antenna Aperture Selection (MAAS)} for joint reception (JR)-CoMP
(see also \cite{khalid}).

The optimization of sets of helper cells is an integer program.
Assuming max-ratio combining of received signals and 
relaxing the integer constraints, the optimization problem becomes convex.
We present a distributed algorithm in which each helper cell 
announces an {\em egress price},
indicating the demand for its signals to help other cells, and each assisted
cell computes an {\em ingress price}, indicating the potential improvement
from adding a helper cell.
The prices are used to compute an ordering of users/cells,
which is then used to allocate the helper antennas across assisted cells.
This assignment is iterated with updates for the prices. 

We refer to this algorithm as {\em LiquidMAAS} due to its ability to flexibly allocate help based on network load conditions, and show that it converges 
to the optimal allocation of helper cells across the network.
Numerical results are presented that show that the algorithm converges
in relatively few iterations even for a system with a large number of cells. Furthermore, the gains relative to an {\em a priori} fixed allocation of helper cells can be substantial.

\section{Multi-Antenna Aperture Selection}
\label{num}
\subsection{Problem Setup}
Consider a network of $J$ cells, each with multiple antennas. Each cell $j$
serves a set of users $u(j)$. The uplink signal from each user $k \in u(j)$
is typically strongest at its serving cell; however, the signal
could also be received with significant strength at other cells
depending on the user location and network topology.
We assume the antenna-combined signal for a particular user $k \in u(j)$
can be shared with other cells, and is, of course,
provided to that user's serving cell for CoMP combining.

Figure (\ref{cell-topology}) shows an example of a network with 
uplink data sharing. There are four cells with seven users,
where $u(1) = \{1,2\}$, $u(2) = \{3\}$, $u(3) = \{4,5,6\}$, and $u(4) = \{7\}$.
Let ${\cal N}_R(k)$ denote the set of cells with
data that can be requested by the cell serving user $k$
({\em ingress} neighborhood). That typically corresponds
to the set of cells where the user's SINR is above 
a minimum threshold value (usually -10 dB). 
That is, ${\cal N}_R(k)=\{i \in {\cal J}, i \neq \sigma(k): 
S_{i \rightarrow j}^k \ge  S_{\min}\}$, where
$S_{i \rightarrow j}^k$ is the SINR at cell $i$
for a user $k$ in cell $j$ (i.e., $k\in u(j)$), and $\sigma(k)$
is the serving cell for user $k$. Note that ${\cal N}_R(k)$
may not be the actual set of helper cells for user $k$
due to backhaul and aperture constraints.

Also shown in Figure (\ref{cell-topology}) are the possible sharing variables 
along with the inter-connect. Note that cells 1 and 3, 
and cells 2 and 4 do not exchange information. 
The sharing variable 
$x_{i \rightarrow j}^k \in \{0,1\}$ indicates whether or not
cell $i$ shares data for user $k$ in cell $j$
over the backhaul link $i \rightarrow j$.
For the example shown in Figure (\ref{cell-topology}), the
local neighborhood sets for the different users are
\begin{eqnarray}
{\cal N}_R(1) = {\cal N}_R(2) = \{2, 4\};  && {\cal N}_R(3) = \{1, 3\}; \nonumber \\
{\cal N}_R(4) = {\cal N}_R(5) = {\cal N}_R(6) = \{2, 4\}; && {\cal N}_R(7) = \{1, 3\} \nonumber
\end{eqnarray}

The maximum receiver aperture size, or number of helper cells, 
for user $k$ is denoted as $L_R(k)$, and is the number of cells
in the neighborhood ${\cal N}_R (k)$. For a cell $i$, its {\em egress}
neighborhood to cell $j$ is denoted by ${\cal N}_T (i, j)$, and is the set of users in cell $j$ it can potentially help,
ignoring backhaul and aperture constraints. That is:
\[
{\cal N}_T (i, j) = \{k \in u(j): i \in {\cal N}_R(k)\} \ ; \ \ i, j \in {\cal J}, \ i \neq j
\]
For the example in Figure (\ref{cell-topology}), 
$L_R (k)=2  \ \forall k\in u(j), \ \forall j=\{1, 2, 3, 4\}$. 
The egress neighborhoods are:
\begin{align}
\nonumber
{\cal N}_T(1, 2) = \{3 \}; \ {\cal N}_T(1, 4) = \{7\}; & \ {\cal N}_T(2, 1) = \{1, 2\};  \\
\nonumber
{\cal N}_T(2, 3) = \{4, ,5, 6 \}; \ {\cal N}_T(3, 2) = \{3 \}; & \ {\cal N}_T(3, 4) = \{4, 5, 6\};  \\ 
\nonumber 
{\cal N}_T(4, 1) = \{1, 2\}; & \ {\cal N}_T(4, 3) = \{4, 5, 6 \}; 
\end{align}
Note that if there is no possibility of sharing between two cells, then the corresponding egress set is null, {\em i.e.,} ${\cal N}_T(2, 4) = {\cal N}_T(4, 2) = \emptyset
$. The notation is summarized in Table \ref{table}.


\begin{table}[ht]
\caption{Variable Definitions} 
\centering 
\begin{tabular}{|c|c|} 
\hline\hline
${\cal J}$ & Set of all cells in the network, with $J=|{\cal J}|$.  \\ \hline
$u(j)$ & Set of users connected to cell $j$. \\ \hline
$\sigma(k)$ & Serving cell for user $k$. \\ \hline
$x_{i \rightarrow j}^k$ & BW fraction that cell $i$ shares with user $k$ in cell $j$. \\ \hline
$S_{i \rightarrow j}^k$ & SINR experienced by user $k$ at cell $i$. \\ \hline
$\beta_k$ & Bandwidth fraction allocated to user $k$. \\ \hline
$\omega_k$ & Scheduling priority weight for user $k$. \\ \hline
$L_A$ & Maximum number of {\em helper} cells allowed. \\ \hline
${\cal N}_R(k)$ & Ingress neighborhood for user $k$ in cell $j$. \\ \hline
$L_R(k)$ & Maximum aperture of helper cells ($= |{\cal N}_R(k)|$). \\ \hline  
${\cal N}_T(i, j)$ & User egress neighborhood for cell $i$ to cell $j$. \\ \hline
$L_{\bar T}$ & Maximum per-cell egress bandwidth. \\ \hline
\end{tabular}
\label{table}
\end{table}

\begin{figure}[h]
\centerline{\includegraphics[height=70mm]{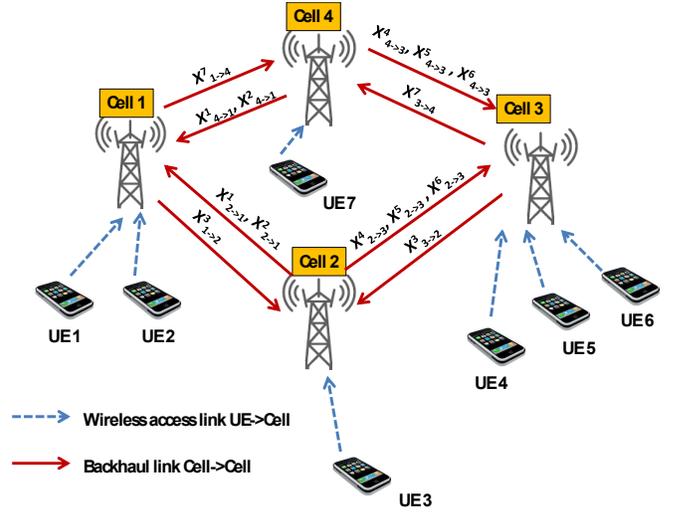}  
}
\caption{Example cell topology with backhaul inter-connects.}
\label{cell-topology}
\end{figure}

\subsection{Network Utility Maximization}
Our problem is to select the set of helper cells for each user $k$
to maximize an overall {\em network} objective, subject to the egress
(backhaul bandwidth) and ingress (aperture) constraints.
This selection is assumed to occur for a fixed
schedule of users in a given transmission time, and after
reception of the signals. 
We assume that a user is always processed by its serving cell's antennas. 
If another cell's antennas are included in the aperture (helper set) for a user,
then that cell performs local processing 
(e.g., Max-Ratio Combining (MRC) or Minimum Mean Squared Error
combining across its own antennas), and forwards those post-combined 
signals to the serving cell. That cell in turn performs MRC 
of the signals from different cells. 
Therefore, the combined SINR is the sum of the SINRs 
from the serving and helper cells. 

We therefore obtain the following optimization problem (NUM)
over the sharing variables $\{x_{i \rightarrow j}^k\}$:
\\ \\
{\bf Maximize}
\begin{eqnarray}
\sum_{j \in {\cal J}} \sum_{k \in u(j)} \omega_k \beta_k \log\left [ 1 + S_{j \rightarrow j}^k + \sum_{i \in {\cal N}_R(k)} \ S_{i \rightarrow j}^k \ x_{i \rightarrow j}^k \right ] && \label{num-cost} 
\end{eqnarray}
{\bf Subject to:}
\begin{align}
\sum_{i \in {\cal N}_R(k)} x_{i \rightarrow j}^k  & \le L_A \ \ \forall j \in {\cal J}, \ \forall k \in u(j) 
\label{ingress-constraint} \\ 
\sum_{j \in {\cal J}, j \neq i} \sum_{k \in {\cal N}_T(i, j)} \ \beta_k \ x_{i \rightarrow j}^k  & \le L_{\bar T} \ \forall i \in {\cal J} 
\label{egress-constraint} \\
x_{i \rightarrow j}^k \in \{0, 1\} \  \forall j \in {\cal J}, \ & \forall k \in u(j), \forall i \neq j 
\label{variable-constraint}
\end{align}
where $w_k$ and $\beta_k$ are defined in Table \ref{table}.

The cost function is concave and the inequality constraints are linear. 
If we relax the 
variables to be $0 \le x_{i \rightarrow j}^k \le 1$, the optimization problem is convex. 
A fractional value for the sharing variable has an interpretation 
of a helper cell sending only part of the resource block (RB) allocated
to a user. Since this is acceptable within the inter-connect architecture, 
and also for the receiver processing, 
we relax the variables to obtain a convex formulation.
In fact, we will show that in the solution to the relaxed problem,
for each user $k$,
all sharing variables except one satisfy the integer constraint.

The ingress constraint is a computational complexity constraint 
based on the maximum number of cell signals that can be processed 
by the receiver, and is independent of the allocated bandwidth fraction 
of the user\footnote{Even ignoring computational limitations,
limiting the number of signals that are combined can improve performance
due to finite training, associated channel estimation error,
and receiver imperfections \cite{khalid}.}. 
The egress constraint, on the other hand, is a backhaul
bandwidth constraint, and hence accounts for the fraction of bandwidth
a user's signal occupies. 


Note that the egress constraint (\ref{egress-constraint}) is the only 
coupling constraint. 
If that were absent, the solution would be straightforward: 
for each user $k \in u(j)$, order the cells according to the SINRs for that user,
and pick the $\min\{L_A, L_R(k)\}$ top cells for receiver processing.

\section{Characterization of Solution}
\label{solution}
We now present a distributed algorithm for solving the preceding NUM problem.
Since the relaxed problem is convex, we seek a set of sharing variables 
that satisfies the KKT conditions.
Define a combined SNR metric for user $k$ as
\begin{equation}
g(\vect{x}_k) = 1 + S_{j \rightarrow j}^k + \sum_{m \in {\cal N}_R(k)} \ S_{m \rightarrow j}^k \ x_{m \rightarrow j}^k
\end{equation}
where $\vect{x}_k = \{x_{i \rightarrow j}^k\}_{i \in {\cal N}_R(k)}$.
We can then write the Lagrangian for the NUM as
\begin{align}
\nonumber
{\cal L}(\vect{x}, \Lambda, \Psi, \Gamma, \Theta)=\sum_{j \in {\cal J}} \sum_{k \in u(j)} & 
{\cal L}_k(\vect{x}_k, \lambda_k, \Psi, \Gamma_k,\Theta_k) \\
 & + \sum_{i \in {\cal J}} \psi_i \ L_{\bar T}
\label{eq:lag}
\end{align}
where the cost function for user $k$ is
\begin{eqnarray}
&& {\cal L}_k (\vect{x}_k, \lambda_k, \Psi, \Gamma_k, \Theta_k) =  
\omega_k \beta_k \log \left [ g(\vect{x}_k) \right ] \nonumber \\
&& +  \lambda_{k} \left (L_A - \sum_{i \in {\cal N}_R(k)} x_{i \rightarrow j}^k \right ) -  \beta_k \ \sum_{i \in {\cal N}_R(k)} \psi_i  \ x_{i \rightarrow j}^k \nonumber \\
&& + \sum_{i \in {\cal N}_R(k)} \gamma_{i,k} \left (1 - x_{i \rightarrow j}^k \right ) + \sum_{i \in {\cal N}_R(k)} \theta_{i,k} x_{i \rightarrow j}^k, \label{lagrangian-peruser}
\end{eqnarray}
the Lagrange multipliers
$\psi_i, \lambda_{k}, \gamma_{i,k}, \theta_{i,k} \ge 0$,
$\{\Gamma_k, \Theta_k\}$ is the set of KKT multipliers per user,
and $\{\Psi, \Lambda, \Gamma, \Theta\}$ is the entire set 
of KKT multipliers. Note that to write the Lagrangian ${\cal L}$
this way, the egress constraint must be broken into partial sums of 
sharing variables for a particular user $k$:
\[
\sum_{i \in {\cal J}} \psi_i  \sum_{j \in {\cal J}, j \neq i} \sum_{k \in {\cal N}_T(i, j)} \beta_k x_{i \rightarrow j}^k \\
= \sum_{j \in {\cal J}} \sum_{k \in u(j)} \beta_k \sum_{i \in {\cal N}_R(k)}  \psi_i x_{i \rightarrow j}^k
\]

The KKT conditions for the NUM are then given by:
\begin{eqnarray}
\frac{\omega_k \beta_k S_{i \rightarrow j}^k}{g(\vect{x}_k)} 
-  \lambda_k - \psi_i \beta_k - \gamma_{i,k} + \theta_{i,k} = 0 && \label{derivative}  \label{KKT} \\
\psi_i \left (L_{\bar T} - \sum_{l \in {\cal J}, l \neq i} \sum_{m \in {\cal N}_T(i, l)} \beta_m \ x_{i \rightarrow l}^m \right ) = 0 && \\
\lambda_k \left (L_A - \sum_{n \in {\cal N}_R(k)} x_{n \rightarrow j}^k \right ) = 0 && \\\
\gamma_{i, j, k} \left (1 - x_{i \rightarrow j}^k \right ) = 0 && \\
\theta_{i,k} x_{i \rightarrow j}^k = 0 && \label{c-slack}
\end{eqnarray}
for every $(i, j) \in {\cal J}, i \neq j, k \in u(j)$. Thus the NUM decomposes into a set of local per-user optimization problems 
for a given a set of egress prices $\{\psi_i\}$.

We can now apply a dual decomposition to obtain a distributed iterative
algorithm for solving the NUM.
Namely, given set of shadow prices $\{\psi_i(t)\}$ and $\{\lambda_k(t)\}$ for the egress and ingress constraints 
(\ref{egress-constraint},\ref{ingress-constraint}) at iteration $t$, each user solves 
its cell selection problem locally by maximizing the objective 
${\cal L}_k(\vect{x}_k, \lambda_k(t), \Psi(t), \Gamma_k,\Theta_k)$. 
Define this {\em primal} solution to be 
$\{\vect{x}_k^*(t), \Gamma_k^*(t), \Theta_k^*(t)\}$. 
Then, the egress and ingress prices can be updated using the sub-gradient algorithm
\begin{eqnarray}
\psi_i(t+1) &=& \left [ \psi_i(t) - \nu \Delta_i(t) \right ]_+ \nonumber \\
\Delta_i(t) &=& \left (L_{\bar T} - \sum_{l \in {\cal J}, l \neq i} \sum_{m \in {\cal N}_T(i, l)} \beta_m \ x_{i \rightarrow l}^{m*}(t) \right ) 
\end{eqnarray}
\begin{equation}
\lambda_k(t+1) = [ \lambda_k(t) - \nu (L_R(k)-\sum_{n \in {\cal N}_R(k)}x_{n \rightarrow j}^{k*}(t)) ]_+
\end{equation}
where $\nu > 0$, is a small step-size, and $[y]_+ = \max\{y, 0 \}$. 


It remains to solve the primal problem in (\ref{lagrangian-peruser}). 
For this define the metric:
\begin{equation}
m_{i,k} = \frac{\beta_k \ \psi_i + \lambda_k}{\omega_k \ \beta_k \ S_{i \rightarrow j}^k}.
\label{metric}
\end{equation}

\begin{lemma} 
\label{lem-1}
If $m_{i,k} \neq m_{l,k}, \ \forall l \neq i$, $\omega_k >0$, $\beta_k > 0$, and
$S_{i \rightarrow j}^k > 0 \ \forall k$, then there can be at most one variable
in $\{x_{i \rightarrow j}^k\}_{i \in {\cal N}_R(k)}$ that has a fractional value. 
All others are either $0$ or $1$. 
\end{lemma}
\begin{proof}
Consider the ingress problem for user $k \in u(j)$, that is, 
maximization of the objective (\ref{lagrangian-peruser}). 
Assuming a fixed set of egress prices $\{\psi_i\}, \ \ i \in {\cal N}_R(k)$, 
we rewrite the KKT condition \eqref{derivative} as
\begin{equation}
\frac{1}{g(\vect{x}_k)} = m_{i,k} + \delta_{i,k} 
\label{key-opt-eqn}
\end{equation}
$\forall i \in {\cal N}_R(k)$ since $\omega_k \beta_k S_{i \rightarrow j}^k > 0$, 
and where 
\begin{equation}
\delta_{i,k} = \frac{\gamma_{i,k}-\theta_{i,k}}{\omega_k \beta_k S_{i \rightarrow j}^k}.
\end{equation}
It follows from the KKT conditions that
\begin{itemize}
\item if $x_{i \rightarrow j}^k = 1$, then $\gamma_{i,k} > 0$, $\theta_{i,k} = 0$, 
and $\delta_{i,k} > 0$; 
\item if $x_{i \rightarrow j}^k = 0$, then $\gamma_{i,k} = 0$, $\theta_{i,k} > 0$, 
and $\delta_{i,k} < 0$;
\item if $x_{i \rightarrow j}^k \in (0,1)$, then $\gamma_{i,k} = \theta_{i,k} = \delta_{i,k} = 0$.
\end{itemize}
 There cannot be two sharing variables
$x_{i \rightarrow j}^k$ and $x_{l \rightarrow j}^k, i \neq l$ that both
take fractional values in the optimal solution, since (\ref{key-opt-eqn}) 
cannot be satisfied for both $i$ and $l$ with $\delta_{i,k} = \delta_{l,k} = 0$ given that $m_{i,k} \neq m_{l,k} \ \ \forall \ l \neq i ;  \ i, l \in {\cal N}_R(k)$. 
\end{proof}
 
Next we characterize the solution to the primal problem for user $k$. 

\begin{definition}
\label{io_def}
 Let $\vect{x}_k^*$ be the optimal solution to \eqref{lagrangian-peruser}. Define the set of active variables in this solution to be ${\cal I}_{\mbox{{\small active}}}$, and the index of the cell with the maximum metric in this set as $i_o$: 
\begin{eqnarray}
{\cal I}_{\mbox{{\small active}}} &=& \{i \in {\cal N}_R(k): 0 < x_{i \rightarrow j}^{k*} \le 1 \} \nonumber \\
i_o &=& \arg \max_{i \in {\cal I}_{\mbox{{\small active}}}} m_{i,k} \label{io-def} 
\end{eqnarray}
\end{definition}

\begin{therm} 
\label{lem-2}
For a fixed $\lambda_k$ and $\psi_i$, $i \in {\cal N}_R(k)$, assume {\em w.l.o.g.} that $m_{i,k}, \{i=1, 2, 3, \ldots, L_R(k)\}$, is increasing for user $k$. Then, the sharing variables that maximize ${\cal L}_k$ in \eqref{lagrangian-peruser} are given by:
\begin{itemize}
\item $x_{i \rightarrow j}^{k*} = 1$ for $i < i_o$,
\item $x_{i \rightarrow j}^{k*} = 0$ for $i > i_o$,
\item $x_{i_o \rightarrow j}^{k*} \in (0, 1)$ if $\frac{1}{1 + S_{j \rightarrow j}^k + \sum_{m \le i_o} \ S_{m \rightarrow j}^k} > m_{i_o,k}$, else $x_{i_o \rightarrow j}^{k*}=1$. 
\end{itemize}
\end{therm}
\begin{proof}
Let the optimal shadow price difference be $\{\delta_{i,k}^*\}$ corresponding to the optimal solution $\vect{x}_k^*$. 
We show that ${\cal I}_{\mbox{{\small active}}}=\{1, 2, \ldots, i_o\}$. 

{\em PART I}: ${\cal I}_{\mbox{{\small active}}} \supseteq \{1, 2, \ldots, i_o\}$:

Since $m_{i,k}$ is increasing in $\{i=1, 2, 3, \ldots, L_R(k)\}$, it follows that to satisfy (\ref{key-opt-eqn}), 
\begin{equation}
\delta_{1,k}^* > \delta_{2,k}^* > \ldots > \delta_{i_o,k}^* > \ldots > \delta_{L_R(k),k}^* \label{delta-seq}
\end{equation}
Since $i_o \in {\cal I}_{\mbox{{\small active}}}, \delta_{i_o,k}^*  \ge 0$. 
Hence, from (\ref{delta-seq}), 
\begin{eqnarray}
\delta_{i,k}^* > 0 &\Rightarrow  & \  x_{i \rightarrow j}^{* k} = 1 \ \ \forall \ i = 1, 2, \ldots, i_o-1 \nonumber \\
\Rightarrow \{1, 2, \ldots, i_o\} &\subseteq & {\cal I}_{\mbox{{\small active}}}
\end{eqnarray}

{\em PART II}: ${\cal I}_{\mbox{{\small active}}} \subset \{1, 2, \ldots, i_o\}$:

For this, we show that $\{i_o+1, i_o+2,\ldots, L_R(k)\} \subseteq {\cal I}_{\mbox{{\small active}}}^c$, that is, $x_{ i \rightarrow j}^{* k} = 0 \ \forall i = i_o+1, \ldots L_R(k)$. This follows by contradiction. Assume that $i' \in {\cal I}_{\mbox{{\small active}}}$, and $i' > i_o$. By definition then $\delta_{i',k}^* \ge 0$. Since $i' > i_o$, we have $\delta_{i',k}^* < \delta_{i_o,k}^*$ from (\ref{delta-seq}). This implies $m_{i',k} > m_{i_o,k}$, which contradicts the definition of $i_o$ in (\ref{io-def}). 

From these, it follows that:
\[
{\cal I}_{\mbox{{\small active}}}= \{1, 2, \ldots, i_o\} 
\]
Further, the one variable that can have a fractional value is $x_{i_o \rightarrow j}^{* k}$ and found as a solution to $\frac{1}{1 + S_{j \rightarrow j}^k + \sum_{m < i_o} \ S_{m \rightarrow j}^k + S_{i_o \rightarrow j}^k x_{i_o \rightarrow j}^k} = m_{i_o,k}$.
\end{proof}

\begin{remark}
\label{rem-1} 
The value of $i_o$ can be then found by a search such that the KKT criterion in (\ref{key-opt-eqn}) is met with $\delta_{i,k}^* > 0$ for $i < i_o$, $\delta_{i,k}^* < 0$ for $i > i_o$ and $\delta_{i_o,k}^* \ge 0$ with some $0 < x_{i_o \rightarrow j}^{* k} \le 1$. 
\end{remark}

\section{LiquidMAAS Algorithm}
\label{liquid-maas}
The following algorithm is based on the preceding properties
of the solution.


\begin{enumerate}
\item {\bf Initialize egress (bandwidth demand) and ingress (aperture) prices}: 
\begin{enumerate}
\item $\psi_i = \varepsilon > 0 \ \forall i$
\item $\lambda_k = 0 \ \forall k \in u(j)$, $j \in {\cal J}$. 
\end{enumerate}
\item For each $k \in u(j)$, $j \in {\cal J}$, do:
\begin{enumerate}
\item {\bf Order helper cells:} For user $k$, order the helper cells in increasing value of metric $m_{i,k}$.
Let that ordered list be ${\cal I}_k = \{i_1, i_2, \ldots, i_{L_R(k)}\}$. 
\item {\bf Select aperture:} Initialize $x_{i \rightarrow j}^k = 0$ for $i= 1, \cdots , L_R (k)$. Add helper cells sequentially,
i.e., for $n=1, 2, \cdots$ add a new helper cell $n$ 
by setting $x_{n \rightarrow j}^k = 1$, and check if
\begin{equation}
\frac{1}{ \left [ 1 + S_{j \rightarrow j}^k + \sum_{p = 1, \ldots, n} \ S_{i_p \rightarrow j}^k  \right ]} > m_{i_n,k}
\label{aperture-selection}
\end{equation}
If true, then increment the aperture size to $n$. 
If false, then stop incrementing $n$ and set
\begin{equation}
x_{i_n\rightarrow j}^k = \max \left \{0, y \right \}
\end{equation}
where
\begin{equation}
y = \frac{\frac{1}{m_{i_n,k}} - [1 +  S_{j \rightarrow j}^k + \sum_{l \le n-1} \ S_{i_l \rightarrow j}^k]}{S_{i_n \rightarrow j}^k} \nonumber 
\end{equation}
and $y \leq 1$, corresponding to the fractional sharing variable.
\end{enumerate}
\item {\bf Update Ingress Price:}  
\[
\lambda_k \leftarrow \max\left(\lambda_k - \nu  [L_R(k) - \hat{L}_R(k)], 0\right)
\]
where $\hat{L}_R(k) = \sum_{i \in {\cal N}_R(k)} x_{i \rightarrow j}^k$, and $\nu > 0$ is a small step-size.
\item {\bf Update demand headroom:} Each cell computes
the egress demand headroom:
\begin{equation}
\Delta_i = L_{\bar T} - \sum_{l \in {\cal J}, l \neq i} \sum_{m \in {\cal N}_T(i, l)} \beta_m \ x_{i \rightarrow l}^{m}
\label{delta}
\end{equation}
\item {\bf Update demand prices:} 
\[
\psi_i \leftarrow \max(\psi_i - \nu \Delta_i, 0), 
\]  
\item {\bf Iterate:} Repeat steps (2)-(5) until convergence (i.e., all $|\hat{L}_R(k) - L_R(k)| < \varepsilon_1$ and $|\Delta_i| < \varepsilon_2$ for some sufficiently small $\varepsilon_1, \varepsilon_2 > 0$).   
\end{enumerate}

Given the tolerances $\varepsilon_1, \varepsilon_2 \rightarrow 0$, 
if $m_{i,k} \neq m_{l,k} \ \forall l \neq k$, $\omega_k >0, \beta_k > 0, S_{i \rightarrow j}^k > 0 \ \forall k$, then the LiquidMAAS algorithm
is guaranteed to converge, and the allocation
approaches the solution to the NUM problem in 
(\ref{num-cost}-\ref{variable-constraint}) as 
$\varepsilon_1, \varepsilon_2 \rightarrow 0$. 
This follows from the results in Lemma \ref{lem-1}, Theorem \ref{lem-2} and the subsequent remarks.
Furthermore, we note that the algorithm can identify the integer-valued
sharing variables in the solution within a finite number of steps
(provided that $\varepsilon_1, \varepsilon_2$ are sufficiently small).
The algorithm can be modified so that this optimality property
also applies to the case
where $m_{i,k} = m_{l,k}$ for some $l \neq k$; however, we omit the details here.


%

The algorithm is referred to as LiquidMAAS since it distributes the helper cell load across the network, ensuring that all cells are able to meet their egress bandwidth constraint, while maximizing network utility within the ingress (aperture) constraint. In contrast, if we do the traditional cell selection approach that considers only SINRs at different cells, there could be situations where the demand on certain cells could be very high, much exceeding the egress bandwidth. If a cell is in high demand {\em i.e.,} many cells request its antenna data, that cell's $\psi_i$ value will be high, indicating that there is a high price to be paid for getting that cell's antenna data. If there is little demand for a cell $i$'s data, then $\psi_i \approx 0$.

If all cells have the same demand, then the cell selection amounts to ordering purely by SINR. If the ingress price $\lambda_k \gg \psi_i \ \forall i$, then the ordering is approximately according to SINR. This ordering takes the potential help one can get from a helper cell in terms of SINR ($S_{i\rightarrow j}^k$) and weighs it against the price of obtaining it ($\psi_i$) and how tight the ingress constraint is ($\lambda_k$). So if we have two cells that give about the same help in terms of SINR, one would pick the cell that is less ``busy" or in demand (smaller $\psi_i$). 
If a cell is heavily in ``demand", some cells may choose to request from other cells that are in less demand, while still getting a reasonable improvement in SINR. 


The algorithm also automatically takes a user's priority ($\omega_k$) into consideration while making a decision on the user's aperture. This priority usually arises out of enforcing fairness for the different users via a utility such as proportional fairness. If a user is deemed to be high priority (large $\omega_k$), then the aperture selection in (\ref{aperture-selection}) becomes small for that user, biasing the selection of larger apertures for such users. This is in agreement with what we expect intuitively, {\em i.e.,} that high priority users should be able to pick larger apertures so as to obtain improved throughput.

\section{Simulation Results}
\label{sims}

We evaluated the LiquidMAAS algorithm for the case of a standard 57-cell layout with an average of $10$ users per cell. The simulation parameters are shown in Table \ref{table2}:
\begin{table}[t]
\renewcommand{\arraystretch}{1.3}
\caption{List of Simulation Parameters}
\label{table:sim_param}
\centering
\begin{tabular}{p{3.2cm}||p{4.3cm}}
\hline
\bfseries Simulation Parameter & \bfseries Values \\
\hline\hline
UL system bandwidth & 10 MHz\\ 
Carrier frequency & 2.0 GHz \\
Load average & 10 UE per cell \\
Inter-site distance & 100m \\
Receiver & MRC \\ 
Number of antennas per cell & 2 \\
Power control ($\alpha$, $Po$) & (0.8, -80 dBm)\\ 
Max transmit power & $24$ dBm \\
Noise figure & 4 dB  \\
Traffic model & Full buffer \\ 
\hline
\end{tabular}
\label{table2}
\end{table}
 
Users were allocated bandwidth equally based on the number of users connected to a given cell. A user's SINR to all cells in the network was calculated based on this bandwidth allocation, and open-loop fractional power control  \cite{lte_814} was assumed. The update for $\psi_i$ and $\lambda_k$ used a step-size of $\nu = 0.005$. Four algorithms were evaluated: {\em (1) No CoMP, (2) MAAS without egress constraints, (3) MAAS with a randomized egress bandwidth control mechanism, and (4) LiquidMAAS}. In case {\em (2)}, each cell requests help based on their aperture limit, and helper cells grant any requested help. In {\em (3)}, helper cells randomly grant requests for help until their egress constraint limit is reached.  

The ingress aperture constraint limit was set to $L_A = 3$, so that each user being gets a maximum of $3$ helper cells' data. The egress bandwidth limit $L_{\bar T}$ was varied to see the effect on performance and convergence. The top part of Figure  (\ref{fig2}) shows the convergence behavior of all cells' egress bandwidth demand (for $L_{\bar T}=1.0$), from which we observe that the algorithm converges in $\approx 50$ iterations. The convergence time is dependent on the step-size for updating the prices. It was observed that if the aperture and egress limits are similar, the step-size should be reduced to enable smoother convergence. The bottom part of Figure (\ref{fig2}) shows the distribution of weighted sum rate (WSR) gain obtained by the various CoMP approaches. This shows that LiquidMAAS gives substantial gain over 'No CoMP' and also over the randomized egress bandwidth control strategy. 

Figure (\ref{fig3}) shows the distribution of the egress bandwidth and ingress apertures at convergence. We observe that the LiquidMAAS algorithm maintains tight control over the egress bandwidth, and  adapts the ingress aperture to ensure that the constraints are met. Finally, Figure (\ref{fig4}) shows the average WSR gain (over 'No CoMP') as a function of the egress limit ($L_{\bar T}$). LiquidMAAS gives significant gains across the range of egress limits compared to the randomized strategy, and converges to the unconstrained MAAS solution for larger values of $L_{\bar T}$.



\begin{figure}[h]
\centerline{\includegraphics[height=95mm]{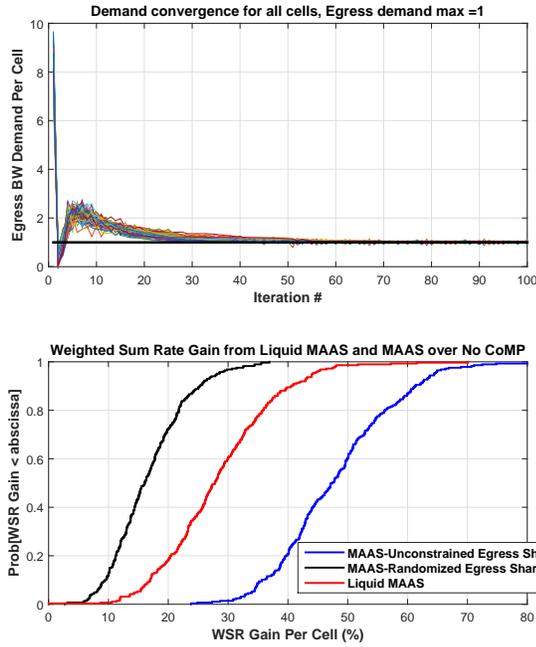}  
}
\caption{MAAS algorithm performance comparison and LiquidMAAS egress bandwidth convergence for $L_{\bar T} = 1.0$.}
\label{fig2}
\end{figure}

\begin{figure}[h]
\centerline{\includegraphics[height=95mm]{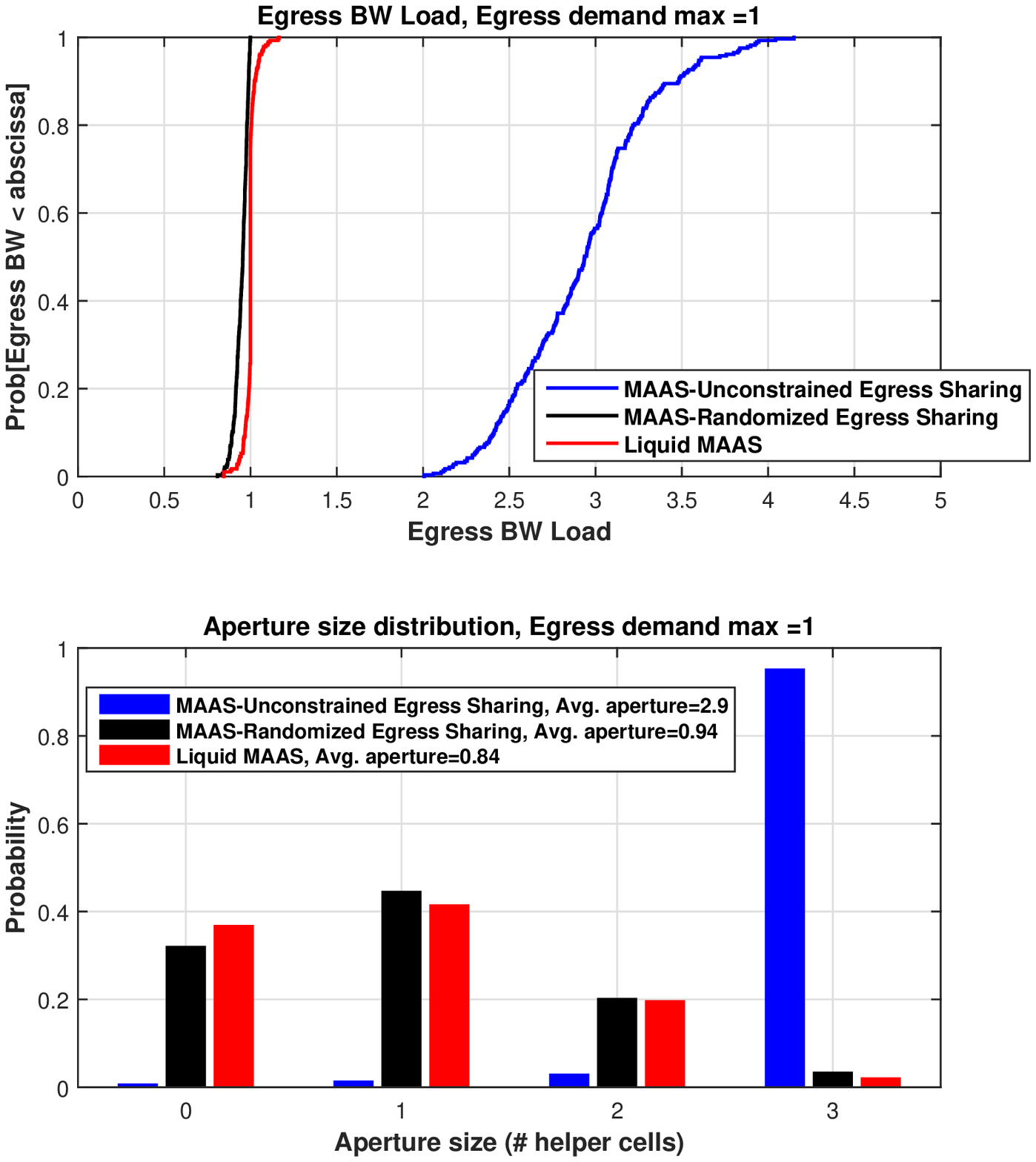}  
}
\caption{LiquidMAAS ingress aperture and egress bandwidth distribution for $L_{\bar T} = 1.0$.}
\label{fig3}
\end{figure}

\begin{figure}[h]
\centerline{\includegraphics[height=65mm]{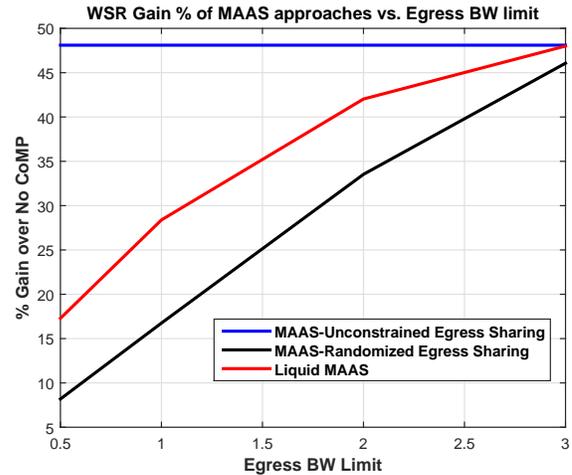}  
}
\caption{Algorithm performance comparison vs. egress bandwidth limit.}
\label{fig4}
\end{figure}




\section{Conclusions}
\label{conc}
In this paper, we presented an approach for joint algorithm and architecture optimization for co-operative communication networks. We considered the problem of uplink joint reception CoMP with backhaul bandwidth constraints, and showed that the resulting network utility maximization problem is convex. A distributed algorithm to solve this problem was presented, which consisted of local nodes computing their desired helper requests, followed by helper cells computing and updating an egress price for their bandwidth. At convergence, these iterations give a set of connections between helper and recipient cells, such that egress and ingress constraints are met at all cells, and overall network utility is maximized. Simulation results for a 57-cell system were presented to illustrate the efficacy of the approach, 
and show that the distributed dual decomposition-based gradient algorithm converges within tens of iterations.

\section*{Acknowledgements}
The authors acknowledge useful discussions with R. Agrawal, M. R. Raghavendra, P. Rasky, and C. Schmidt.

\end{document}